\documentclass[review]{elsarticle}

\usepackage{amsmath,amsthm,amsfonts,amssymb}
\usepackage{vhistory}

\theoremstyle{plain}
\newtheorem{theorem}{Theorem}[section]

\newtheorem{proposition}[theorem]{Proposition}

\theoremstyle{definition}

\newtheorem{remark}[theorem]{Remark}

\usepackage[left=1.1in,right=1.1in,top=1in,bottom=1in]{geometry}

\newcommand{\E}{{\mathbb{E}}}
\newcommand{\N}{{\mathbb{N}}}

\newcommand{\Q}{{\mathbb{Q}}}
\newcommand{\R}{{\mathbb{R}}}
\newcommand{\C}{{\mathbb{C}}}

\newcommand{\BTau}{\mathcal{O}}

\newcommand{\Z}{{\mathbb Z}}

\newcommand{\be}{\begin{equation}}
\newcommand{\ee}{\end{equation}}
\newcommand{\bea}{\begin{eqnarray}}
\newcommand{\eea}{\end{eqnarray}}
\newcommand{\beast}{\begin{eqnarray*}}
\newcommand{\eeast}{\end{eqnarray*}}
\newcommand{\bproof}{\begin{proof}}
\newcommand{\eproof}{\end{proof}}

\def\e{\mathrm{e}}

\begin{document}

\begin{frontmatter}

\title{Option Pricing with Legendre Polynomials}


\author{Julien Hok \footnote{Corresponding author}}
\address{Credit Agricole CIB, Broadwalk House, 5 Appold St\\
		 London, EC2A 2DA, United Kingdom\\
		\it{julienhok@yahoo.fr}}

\author{Tat Lung (Ron) Chan \footnote{Senior Lecturer}}
\address{University of East London, Water Lane\\
Stradford, E15 4LZ, United Kingdom} 

\begin{versionhistory}
  \vhEntry{1.0}{10.10.2016}{}{1st version}
  \vhEntry{1.1}{19.03.2017}{}{revised after review}
\end{versionhistory}

\begin{abstract}
Here we develop an option pricing method based on Legendre series expansion of the density function. The key insight, relying on the close relation of the characteristic function with the series coefficients, allows to recover the density function rapidly and accurately. 
Based on this representation for the density function, approximations formulas for pricing European type options are derived. To obtain highly accurate result for European call option, the implementation involves integrating high degree Legendre polynomials against exponential function.  Some numerical
instabilities arise because of serious subtractive cancellations in its formulation (\ref{IntLegendreExp1}) in proposition \ref{PropIntLegendreExp}. To overcome this difficulty, we rewrite this quantity as solution of a second-order linear difference equation and solve it using a robust and stable algorithm from Olver.
Derivation of the pricing method has been accompanied by an error analysis.   
Errors bounds have been derived and the study relies more on smoothness properties which are not provided by the payoff functions, but rather by the density function of the underlying stochastic models. 
This is particularly relevant for options pricing where the payoffs of the contract are generally not smooth functions. The numerical experiments on a class of models widely used in quantitative finance show exponential convergence.
\end{abstract}


\end{frontmatter}


\section{Introduction}

In option pricing, Feynman-Kac formula \cite{KarShreve97} establishes a link between the conditional expectation of the value of a contract payoff function under the risk-neutral measure and the solution of a partial differential equation. In the research areas covered by this theorem, various numerical pricing techniques can be developed. Existing numerical methods can be classified into three major groups: partial integro-differential equations methods, Monte Carlo simulations and numerical integration methods. Each of them has its advantages and disadvantages for different financial models and specific applications. In this paper, we concentrate on the last group for the pricing of European type option.

The point-of-departure for pricing European option with numerical integration techniques is the risk-neutral valuation formula:
\begin{equation}\label{riskneutralvaluation}
 V(x, t_0 = 0) = e^{-rT} \E_{\Q} [V(S_T, T) |S_0 = x] = e^{-rT} \ \int_{\R} V(y, T)\tilde{f}(y |x)dy
\end{equation}

with $\E_{\Q}$ the expectation operator under risk-neutral measure $\Q$, $S_t$ the underlying asset price at $t$ and $T$ the option maturity. $V(x,t)$ denotes the option value at $t$ with $x$ the state variable. $\tilde{f}(y |x)$ is the probability density function of $S_T$ given $S_0=x$ and $r$ the risk-free interest rate.

Unfortunately, for many relevant pricing processes, their probability densities are usually unknown. 
On the other hand, the Fourier transform of these densities, i.e, the characteristic functions, are often available. For instance, from the Levy-Khinchine theorem \cite{ContTankov04} the characteristic functions of Levy processes are known. Or characteristic functions have been derived in the pure diffusion context with stochastic volatility \cite{Heston93} and with stochastic interest rates \cite{BakChen97}. Hence, the Fourier transform methods for option pricing have been naturally 
considered by many authors (see \cite{CarrMadan99} and references therein). Subsequently, some new numerical methods
are proposed. For example, The quadrature method (QUAD) method was introduced by Andricopoulos et al \cite{AndWidDu03}, the Convolution method (CONV) was presented by Lord et al \cite{LordFangBer08}. A fast Hilbert transform approach was considered by 
Feng and Linetsky \cite{FengLinet08}.  The highly efficient Fourier-cosine series (COS) technique, based on Fourier-cosine series expansion of the density function, was proposed by Fang and Oosterlee \cite{FangOosterlee08} and has generated  other developments by Hurn et al \cite{HurnLindsayClelland13} or by Ding et al \cite{DingU11}. Recently, Necula et al \cite{NeculaDriFar16} have employed the modified Gram-Charlier series expansion, known as the Gauss-Hermite expansion, for the density function and obtained a closed form pricing formula for European option. 

In this manuscript, we consider an alternative and propose to expand the probability density function $\tilde{f}(y)$, restricted on a finite interval $[a,b]$, using Legendre polynomials when the characteristic function is known. 
For approximating non periodic function on a finite interval, among the class of basis functions, it is usually recommended to use either Legendre polynomials or Chebyshev polynomials (see page 510 table A.1 in \cite{Boyd00}).  
Legendre polynomial offers tractability property  allowing to compute analytically many quantities of interests.
For example, Legendre polynomial has an analytical formula  for its Fourier transform as in (\ref{FourierTransform}), which is instrumental and used to recover the coefficients $A_n$ in the series expansion of the density function (\ref{ThLegendreSeries}).
The Fourier transform for Chebyshev polynomials does not have a simple closed form and requires some numerical 
approximations (see discussion in \cite{EvansWebster99}).  Moreover the experiments show this formula is numerically stable for large $n$. Generally, the classical Legendre series offers the simplest method of representing a function using polynomial expansion means \cite{Fishback07}. Also the recent analysis by Cohen and Tan \cite{CohenaTan12} shows Legendre polynomial approximation yields an error at least an order of magnitude smaller than the analogous Taylor series approximation and the authors strongly suggest that Legendre expansions, instead of Taylor expansions, should be used when global accuracy is important. Finally, polynomials are convenient to manipulate in general  and we compute simply the European option pricing formula by integrating the payoff against Legendre polynomial functions.

Adrien Marie Legendre, a French mathematician who discovered the famous
polynomials, was never aware of that how much it will be used in developing mathematics. This
Legendre polynomial is being used by mathematicians and engineers for variety of mathematical and numerical
solutions. For example, in physics, Legendre and Associate Legendre polynomials
are widely used in the determination of wave functions of electrons in the orbits of an atom \cite{DickeWittke60, Hollas92} and in the determination of potential functions in the spherically symmetric geometry \cite{Jackson62}. In numerical analysis, Legendre polynomials are used to efficiently calculate numerical integrations by Gaussian quadrature method \cite{MugYeIq06}. 
 
Legendre polynomials is not widely used in quantitative finance but not new. For example, Pulch et al  \cite{PuchEm09} consider the fair price of options as the expected value of a random field where the input volatility parameter is written as a linear function of uniform random variable. The Polynomial chaos theory using Legendre polynomial yields an 
efficient approach for calculating the required fair price. Or in \cite{IbsenAlmeida05}, the authors develop arbitrage free interest rate models for a family of term structures parametrized by linear combinations of Legendre polynomials. Each polynomial provides a clear interpretation in terms of the type of movements that they generate for the term structure (see also \cite{AlmeidaDuarteFer98} and \cite{Almei04b}).

 
To our knowledge, it is the first time that Legendre polynomials are used to expand the probability density function of asset prices and option pricing. To recover rapidly and accurately the density function, our key insight relies on the close relation of the characteristic function with the series coefficients of the Legendre polynomials expansion of the density function (see our result in theorem \ref{PropDensityProxy}). Based on this representation for the density function, approximations formulas for pricing European type options are derived. To obtain highly accurate result for European call option, the implementation involves integrating high degree Legendre polynomials against exponential function. Some numerical
instabilities arise because of serious subtractive cancellations in its formulation (\ref{IntLegendreExp1}) in proposition \ref{PropIntLegendreExp}. To overcome this difficulty, we rewrite this quantity as solution of a second-order linear difference equation in proposition (\ref{SndOrderDiffEq}). To solve this equation in a stable way, we use Olver's algorithm which allows to evaluate these quantities to machine accuracy. Then we develop an analysis to provide estimations of the errors. We believe that a rigorous error estimate is of first importance because the accuracy of our expansion formulas depends on the regularity of the density function. Once done, it brings confidence in the derived expansion and sheds light on the needed assumptions (see our results in propositions \ref{PropositionEpsilon2}, \ref{proposition:epsilon3} and \ref{PropositionEpsilon4}).        

This paper is structured as follows. In section 2, we develop the series expansion of the density function using Legendre polynomials. Based on this, we derive, in section 3, the formulas for pricing European type options and propose a robust and stable procedure for the implementation. An error analysis is presented in section 4. Some numerical experiments are given in section 5. The final section concludes.

\section{Series expansion of density function with Legendre polynomials}
The objective is to estimate the density function $\tilde{f}(y)$ using Legendre polynomials given its characteristic function.
\subsection{Generalized Fourier Series-Legendre Polynomials}

The Legendre polynomials $(P_n(t))_{n \geq 0}$ form a complete basis over the interval $[-1,1]$ and 
can be defined, in term of power series, by 
\begin{equation}\label{LegendrePowerSeries}
P_n(t) = \frac{1}{2^n} \sum_{k=0}^{\lfloor \frac{n}{2} \rfloor} (-1)^k C_n^k C_{2n-2k}^n  t^{n-2k}
\end{equation} 
with $\lfloor r \rfloor$ the floor function and $C_n^k = \frac{n!}{k!(n-k)!}$ the binomial coefficients \cite{Lebedev72, Davis75}.\\

The Legendre basis polynomials can be generalized to cover an arbitrary interval
$[a, b]$ by the change of variable $t = \frac{(2x - (a + b))}{(b - a)}$
which leads to the following

\begin{equation}\label{LegendrePowerSeriesNormalized}
P_n(x) = \frac{1}{2^n} \sum_{k=0}^{\lfloor \frac{n}{2} \rfloor} (-1)^k C_n^k C_{2n-2k}^n  \left[\frac{(2x - (a + b))}{(b - a)}\right]^{n-2k}.
\end{equation} 

Sturm-Liouiville theory guarantees the orthogonality of Legendre polynomials
and it also shows that we can represent functions on $[a, b]$ as a linear combination
of Legendre Polynomials. Thus for suitable $f(x)$ on $[a, b]$ we have the generalized Fourier series

\begin{equation}\label{LegendreSeries}
f(x) = \sum_{n=0}^{\infty} A_nP_n \left( \frac{2x-(a+b)}{b-a} \right)
\end{equation}

where $\{ A_n \}_{n=0}^{\infty}$ is a set of coefficients. To find each $An$, we use the orthogonality relation 

\begin{equation}
\int_a^b P_n\left( \frac{2x-(a+b)}{b-a} \right)P_m \left( \frac{2x-(a+b)}{b-a}\right) dx=\delta_{n=m} \frac{(b-a)}{2m+1}
\end{equation}

and then multiply both sides of
expression (\ref{LegendreSeries}) by $P_m \left( \frac{2x-(a+b)}{b-a} \right)$ and integrate to obtain

\begin{align}
\int_a^b f(x)P_m\left( \frac{2x-(a+b)}{b-a}\right)dx &= \sum_{n=0}^{\infty} A_n\int_a^b P_n \left( \frac{2x-(a+b)}{b-a} \right) P_m \left( \frac{2x-(a+b)}{b-a}\right)dx\\
&= (b-a) \frac{A_m}{2m+1}.
\end{align}

so that 

\begin{equation}\label{Ancoefficients}
A_n = \frac{2n+1}{b-a} \int_a^b f(x)P_n\left( \frac{2x-(a+b)}{b-a}\right)dx
\end{equation}

\subsection{Approximate risk-netural probability density function using standard Fourier series}

we briefly revise the definition of complex Fourier series \cite{Tolstov62, Davis75}. For a
suitable function $f(t)$ supported on $[-\pi, \pi]$, the complex Fourier series representation is given by

\begin{equation}
f(t) = \sum_{k=-\infty}^{+\infty} B_k e^{ikt}, \, \, \, with\, \, \, B_k = \frac{1}{2\pi} \int_{- \pi}^{\pi} f(t)e^{-ikt}dt.
\end{equation}

If we extend the series to support function with a finite range of $[a, b]$, the
complex Fourier series expansion can be defined as:

\begin{equation}\label{fComplexFourier}
f(x) = \sum_{k=-\infty}^{\infty} B_ke^{i(\frac{2\pi}{b-a}x)k}, \,\,\, with \,\,\, 
B_k = \frac{1}{b-a}\int_a^b f(x)e^{-ik(\frac{2\pi}{b-a}x)}dx.
\end{equation}

The formula is achieved through use change of variables:

\begin{equation}
x = \frac{b-a}{2\pi}t + \frac{a+b}{2} \,\, or \,\, t = \frac{2\pi}{b-a}x - \frac{\pi (a+b)}{b-a}
\end{equation}

Being given a probability density function $\tilde{f}(x)$ and its characteristic function $\varphi(u)$, these two functions form a Fourier pair:
\begin{align}
\varphi(u) = & \int_{\R}e^{iux}\tilde{f}(x)dx \\
\tilde{f}(x) = & \frac{1}{2\pi}\int_{\R} e^{-iux}\varphi(u)du.
\end{align}

A necessary condition for $\tilde{f}(x)$ to be a probability density function is that $\tilde{f}(x) \to 0$ as 
$\mid x \mid \to \infty $, and therefore there is guaranteed to be an interval $[a,b]$ such that for all
$x \in ( -\infty, a] \cup [b, \infty)$ it can be asserted that $\tilde{f}(x) < \epsilon$ for any arbitrary small positive $\epsilon$.\\
Let's consider $f(x)$ as the restriction of $\tilde{f}(x)$ on $[a,b]$. We shall discuss the appropriate choice of $[a, b]$ in section \ref{subsection:Truncation Range}.

From (\ref{LegendreSeries}) and (\ref{fComplexFourier}), $f(x)$ can be expressed either in a complex Fourier series or a Legendre polynomials series. As the aim of this paper is to apply Legendre polynomials for a pricing formula, we show how one can precisely approximate $f(x)$ with Legendre series and formulate the coefficients in the expansion knowing the characteristic function. To achieve this, we use (\ref{fComplexFourier}) in (\ref{Ancoefficients}) and assume we can change the order of integration to write

\begin{equation}\label{AnChangeIntegrationOrder}
A_n =  \frac{2n+1}{b-a} \sum_{k = -\infty}^{+\infty} B_k \int_a^b P_n \left( \frac{2x-(a+b)}{b-a} \right) e^{i2\pi(\frac{xk}{b-a})} dx
\end{equation}

Through change of variables, $x = \frac{b-a}{2}t + \frac{a+b}{2}$, and a closed-form expression
for 

\begin{equation}\label{FourierTransform}
\int_{-1}^1 P_n(x)e^{ i \lambda x}dx = i^n \sqrt{\frac{2 \pi}{\lambda}} J_{n+\frac{1}{2}}(\lambda), \, \lambda \in \C
\end{equation}

with $J_{\nu}(z)$ Bessel function of first kind (see \cite{OlverLozierBoiCl10} p.217 and p.456), it comes 

\begin{equation}\label{LegendreFourierCoeff}
\int_a^b P_n\left( \frac{2x-(a+b)}{b-a}\right) e^{i2\pi(\frac{xk}{b-a})} dx = \left\lbrace 
\begin{array}{l}
i^n \frac{(b-a)}{2} 
e^{\frac{i \pi k(a+b)}{b-a}}\sqrt{\frac{2}{k}} J_{n+\frac{1}{2}}(\pi k), \, \, k \neq 0,\\
(b-a)\delta_{n=0}, \, \, k=0. 
\end{array}
\right.
\end{equation}

and so

\begin{equation}\label{SeriesforAn}
A_n = \frac{2n+1}{\sqrt{2}}  \left[ \sum_{k \neq 0} B_k i^n e^{\frac{i \pi k(a+b)}{b-a}}\frac{J_{n+\frac{1}{2}}(\pi k)}{\sqrt{k}} + B_0 \sqrt{2} \delta_{n=0}\right]
\end{equation}

Knowing the characteristic function, we write

\begin{equation} \label{RelationBkBkTildeRk}
B_k = \widetilde{B}_k - R_k
\end{equation} 

with 

\begin{equation}
\widetilde{B}_k := \frac{1}{b-a} \varphi \left( \frac{-2k \pi}{b-a} \right) 
\end{equation}

and 

\begin{equation}\label{Rkdefinition}
R_k := \frac{1}{b-a} \int_{\R -[a,b]} \widetilde{f}(x) e^{-i2\pi(\frac{xk}{b-a})} dx.
\end{equation}

$R_k$ is expected to be small and can be bound as

\begin{equation}
\mid R_k \mid \leq \frac{1}{b-a} \left[  \int_{-\infty}^{a} \widetilde{f}(x) dx + \int_{b}^{ +\infty } \widetilde{f}(x) dx \right] = \frac{1}{b-a} \left[ \widetilde{F}(a) + 1 - \widetilde{F}(b) \right].
\end{equation}

where $\widetilde{F}(x)$ is the cumulative distribution function of $\widetilde{f}(x)$.\\

Finally, using (\ref{RelationBkBkTildeRk}), $A_n$ can be written as

\begin{equation}\label{SeriesforAnFinal}
A_n = \frac{2n+1}{\sqrt{2}}  \left[ \sum_{k \neq 0} \widetilde{B}_k i^n e^{\frac{i \pi k(a+b)}{b-a}}\frac{J_{n+\frac{1}{2}}(\pi k)}{\sqrt{k}} + \widetilde{B}_0 \sqrt{2} \delta_{n=0}\right] + R_{A_n}
\end{equation}

with 

\begin{equation}\label{ExpRAn}
R_{A_n} = -\frac{2n+1}{\sqrt{2}}  \left[ \sum_{k \neq 0} R_k i^n e^{\frac{i \pi k(a+b)}{b-a}}\frac{J_{n+\frac{1}{2}}(\pi k)}{\sqrt{k}} + R_0 \sqrt{2} \delta_{n=0}\right].
\end{equation}

Before summarising the result of the development above in the next theorem, we introduce a couple of definitions and notation taken from \cite{Tolstov62}: a function $f(x)$ is said to be {\it{piecewise smooth}} on the interval $[a,b]$ if either $f(x)$ and its derivative are both continuous on $[a,b]$, or they have only a finite number of {\it{jump discontinuities}} on $[a,b]$. If $x_0$ is a point of discontinuity of a function  $f(x)$ and if the right-hand and left-hand limits exist, 
$x_0$ is said to be a point of {\it{jump discontinuity}}. 
We set 

\begin{equation}
f^n_k(x) = B_k P_n \left( \frac{2x-(a+b)}{b-a} \right) e^{i2\pi(\frac{xk}{b-a})},  \,\, x \in [a,b], \, k \in \Z , \,\, n \in N.  
\end{equation}
 
and consider, for a given $n \in \N$, the series of functions 

\begin{equation} \label{seriesfnk}
\sum_{k= -\infty}^{+\infty} f^n_k(x). 
\end{equation}

\begin{theorem}\label{PropDensityProxy} 
Let's denote by $f(x)$, the restriction of the probability density function $\tilde{f}(x)$ on $[a,b]$ large enough such that $f(a) = f(b)$ and $\varphi(x)$ the characteristic function associated to $\tilde{f}(x)$.
Assume that $f(x)$ is a continuous piecewise smooth function and that the series (\ref{seriesfnk}) is uniformly convergent on $\in [a,b]$ for all $n \in \N$. Then we have the following Legendre series representation

\begin{equation}\label{ThLegendreSeries}
f(x) = \sum_{n=0}^{\infty} A_nP_n \left( \frac{2x-(a+b)}{b-a} \right)
\end{equation}

with $A_n$ given in (\ref{SeriesforAnFinal}). 
\end{theorem}

\begin{proof}

$f(x)$, being continuous and piecewise smooth on $[a,b]$, can be written as in  (\ref{LegendreSeries}) and (\ref{fComplexFourier}) (see e.g \cite{Tolstov62} and \cite{Lebedev72}). The uniform convergence of the series (\ref{seriesfnk}) allows to interchange 
the order of integration in the expression of $A_n$ in  (\ref{AnChangeIntegrationOrder}).

\end{proof}

\remark{ 

\begin{itemize}

\item The representation (\ref{ThLegendreSeries}) allows to retrieve the density function accurately when the characteristic function $\varphi(x)$ is known by truncating the infinite sums in $n$ of (\ref{ThLegendreSeries}) and in $k$ of (\ref{SeriesforAnFinal}) and neglecting the terms $R_{A_n}$ (see section \ref{sec:Numerical experiments} for illustrations).

\item In quantitative finance, the probability density function $\tilde{f}(x)$ of asset prices tends to be smooth in general.
When analytical formulas are available as for Black Scholes model in equation (\ref{GaussianDerivatives}) and in Merton jump diffusion model in equation (\ref{Mertondensity}), we observe that their density functions are infinitely differentiable.
The Malliavin calculus or the stochastic calculus of variations can be applied to the study of existence and smoothness of density for the solution of a stochastic differential equation (SDE) (see e.g \cite{Nualart06} or \cite{Bally03}).  
\end{itemize}

}

\section{A new computational method for option pricing} \label{sec:optionpricing}

\subsection{Option pricing} \label{subsec:optionpricing}

Here, we show how to evaluate European style options using the asymptotic expansion of the density function obtained previously. We denote the log-asset prices by
\begin{equation}
x:=ln \left( \frac{S_0}{K} \right) \,\, and \, \, y:=ln \left( \frac{S_T}{K} \right),
\end{equation} 

with $S_t$ the underlying price at time $t$ and $K$ the strike price. The payoff for
European options, in log-asset price, reads
\begin{equation}\label{callput}
V(y,T)=[\alpha.K(e^y-1)]^+ \,\, with \,\, \alpha = \left\lbrace 
\begin{array}{lr}
1 & \text{for a call}, \\
-1 & \text{for a put}, \\
\end{array}
\right.
\end{equation}
and 
\begin{equation}\label{digital}
V(y,T)= 1_{\alpha y \geq 0} \,\, with \,\, \alpha = \left\lbrace 
\begin{array}{lr}
1 & \text{for a digital call}, \\
-1 & \text{for a digital put}, \\
\end{array}
\right.
\end{equation}

\vspace{0.3cm}

In the following, we focus on the pricing formula for European call option and European digital call option.
The European put option and European digital put option prices can be deduced by parity. 
Indeed, call/put and digital options are very popular in the financial markets for hedging and speculation.
They are also important to financial engineers as building blocks for constructing more complex option products.
For example, it is well-known that the price of European-type option with twice differentiable payoff can be replicated 
model free by a static portfolio consisting of pure discount bond, at the money European Call and put options and a continuum out of the money European Call and put options (see e.g \cite{Nachman98} and \cite{CarrMadan02}). Moreover, pricing and hedging of digital options are challenging because of payoff discontinuity (see discussions in remark \ref{regularitypayoff} and 
\cite{AvellanedaLaurence99}). So it is instrumental to be able to price these options accurately in a robust way.

With (\ref{riskneutralvaluation}), the European call price is given by
\begin{equation}
V(x,0) = e^{-rT} K\E [(e^y-1)^+] =  e^{-rT} K\int_{- \infty}^{+\infty}(e^y-1)^+ \tilde{f}(y | x)dy 
\end{equation}

Since the density rapidly decays to zero as $y \to \pm \infty$, we truncate the
infinite integration range without loosing significant accuracy to $[a, b] \subset \R$ and obtain the approximation 
\begin{equation}\label{V1Approximation}
V_1(x,0) = e^{-rT} K\int_{a}^{b}(e^y-1)^+ f(y | x)dy. 
\end{equation}
In the second step, we replace $f(y | x)$ by its Legendre series representation (\ref{ThLegendreSeries}) to obtain the following proposition
\begin{proposition}\label{PropositionPricingFormulas}
Under the hypotheses of theorem (\ref{PropDensityProxy}), we obtain an approximation of (\ref{riskneutralvaluation}) given by the following Legendre polynomial pricing formula

\begin{equation} \label{V4Approximation}
V_4(x,0) = e^{-rT} \sum_{n=0}^N A^M_nV_n
\end{equation}

where $A^M_n$ and $V_n$ are defined respectively by 

\begin{equation}\label{CoefficientsA_NM}
A^M_n = \frac{2n+1}{\sqrt{2}}  \left[ \sum_{k = -M, \neq 0}^{M} \widetilde{B}_k i^n e^{\frac{i \pi k(a+b)}{b-a}}\frac{J_{n+\frac{1}{2}}(\pi k)}{\sqrt{k}} + \widetilde{B}_0 \sqrt{2} \delta_{n=0}\right]
\end{equation}

\begin{equation}\label{VnExpression}
V_n = \left\lbrace 
\begin{array}{lr}
K \beta \left[ e^{\frac{a+b}{2}} \int_{\alpha}^1 e^{\beta t}P_n(t)dt - \frac{P_{n-1}(\alpha) - P_{n+1}(\alpha) }{2n+1} \right] & \text{for European call}  \\
& \\
\frac{P_{n-1}(\alpha) - P_{n+1}(\alpha) }{2n+1} & \text{for European digital call}
\end{array}
\right.
\end{equation}

with $\alpha = \frac{a+b}{a-b}$ and $\beta =  \frac{(b-a)}{2}$.
\end{proposition} 

\begin{proof}
For the European call price, we use the representation (\ref{ThLegendreSeries}) of  $f(y | x)$, perform two truncations in the infinite sums: One for $n$ in (\ref{ThLegendreSeries}) to $N$ and an another one for $k$ in (\ref{SeriesforAnFinal}) to $[-M, M]$ and neglect the remaining term $R_{A_n}$ in (\ref{ExpRAn}). Then we get an estimation of the price given by

\begin{equation}
V_4(x,0) = e^{-rT} K \sum_{n=0}^N A^M_n \int_a^b (e^y-1)^+P_n \left( \frac{2y-(a+b)}{b-a} \right) dy.
\end{equation}

Without loss of generality, we suppose $a << 0$, $b >>0$ and with a change of variable $t = \frac{2y-(a+b)}{b-a}$, 
the last expression becomes

\begin{equation}
V_4(x,0) = e^{-rT} K \sum_{n=0}^N A^M_n \beta \left[ e^{\frac{a+b}{2}} \int_{\alpha}^1 e^{\beta t}P_n(t)dt - \int_{\alpha}^1P_n(t)dt \right].
\end{equation}

By using the Legendre polynomial property

\begin{equation}\label{LegendreDerivativeProperty}
(2n+1)P_n(t) = P^{'}_{n+1}(t)-P^{'}_{n-1}(t)
\end{equation}

and $P_n(1) = 1, \, \forall \, n \geq 0$, we get

\begin{equation}
\int_{\alpha}^1P_n(t)dt = \frac{P_{n-1}(\alpha) - P_{n+1}(\alpha) }{2n+1}
\end{equation}

The European digital call price is computed similarly.
\end{proof}

\begin{remark}

\begin{itemize}

\item The computation of $\int_{\alpha}^1 e^{\beta t}P_n(t)dt$ in $V_n$ for the European call price needs attention. We provide an analytical formula in proposition (\ref{PropIntLegendreExp}). Its computation is straightforward for small values of $n$. For $n >>1$ several accuracy and stability issues arise because of rapid accumulation of round-off errors \cite{SeleRaHerPeFer13, KlemmSigLar90}. 

\item The valuation for other contracts like asset-or-nothing options, gap options or standard power options     
\cite{Haug07} can be computed similarly.
\end{itemize}
\end{remark}

\subsection{Alternate computational procedure} \label{subsection:AlternateComputational}
The computation of the Legendre pricing formula (\ref{V4Approximation}) is straightforward for small value of $N$ by using the expression for $V_n$  in proposition (\ref{PropIntLegendreExp}). To obtain accurate pricing, we need to consider $N, M >>1$. $M$ large does not have any implementation difficulty. However, for $N>>1$, the computation of $V_n$  using (\ref{IntLegendreExp1}) introduces instability and inaccuracy issues because of 
cancellations  \cite{KlemmSigLar90}. The objective in this section is to provide an alternative computational procedure for $V_n$ with machine accuracy.\\  

Let's write 
\begin{equation}\label{Un}
U_n = \int_{\alpha}^1 e^{\beta t}p_n(t)dt
\end{equation}

Using integration by parts, we get
\begin{equation}\label{UnIIP}
U_n  = W_n - \frac{1}{\beta} \int_{\alpha}^1 e^{\beta t}p^{'}_n(t)dt
\end{equation}

with $W_n = \frac{1}{\beta} ( e^{\beta} - e^{\beta \alpha} P_n(\alpha))$.\\
With (\ref{LegendreDerivativeProperty}), it is easy to show
\begin{equation} \label{LegendreDerivativeAsLegendrePoly}
p^{'}_{n}(t) = \left\lbrace 
\begin{array}{lr}
\frac{2}{||p_{n-1}||^2} p_{n-1}(t) + 
\sum_{i=0, \, 2i \leq (n-3)}  \frac{2}{||p_{2i}||^2} p_{2i}(t) & \text{for odd} \, n \geq 3, \\
 \frac{2}{||p_{n-1}||^2} p_{n-1}(t) + 
\sum_{i=0, \, 2i+1 \leq (n-3)}  \frac{2}{||p_{2i+1}||^2} p_{2i+1}(t) & \text{for even} \, n \geq 2 \\
\end{array}
\right.
\end{equation}

where  $||p_m||^2  = \frac{2}{2m+1}, \, m \geq 0$.\\
Using  (\ref{UnIIP}) with (\ref{LegendreDerivativeAsLegendrePoly}), we then get 

\begin{equation}\label{Unrecurrence}
U_n = W_n-\frac{2(n-1)+1}{\beta}U_{n-1} + U_{n-2}- W_{n-2}
\end{equation}

given $U_0$ and $U_1$.

The next proposition summarizes the second-order linear difference equation for the computation of $U_n$:

\begin{proposition}\label{SndOrderDiffEq}
By posing $Y_n = U_n-W_n$, $Y_n$ satisfies the following second-order linear difference equation
\begin{equation}\label{ynrecurrence}
Y_{n-1}-\frac{1}{\beta}(2n+1)Y_n - Y_{n+1} = \frac{1}{\beta}(2n+1)W_n.
\end{equation} 
given $Y_0$ and $Y_1$.
\end{proposition}

The computation of $U_n$ or $Y_n$ using these forward recurrences is straightforward  but it generates instabilities and inaccuracies for large $n$. It is a well known issue as discussed in \cite{Olver67, Cash77, Gautschi67}. \\
An excellent technique which evaluates $U_n$ in a stable way to machine accuracy is Olver's method \cite{Olver67}. The approach consists to treat the difference equation as a boundary-value problem rather than using initial-value technique. This rewrites the recurrence relation as a triple of recurrence relations, two of which are
evaluated forwards to an index greater than the desired $N$, the number of
additional steps required for a given accuracy being determined as part of the
procedure. The third relation is then evaluated by backward recurrence (see \cite{Olver67} for details).  

\section{Error analysis} \label{ErrorAnalysis}

First, let's write the successive approximations introduced in the derivation of the pricing formula (\ref{V4Approximation}).
\begin{equation}
V(x,0) = \int_{-\infty}^{\infty}V(y,T) \tilde{f}(y | x)dy = V_1(x,0) + \epsilon_1
\end{equation}
with 
\begin{equation}\label{ErrorAnv1}
V_1(x,0) = \int_{a}^{b}V(y,T) \tilde{f}(y | x)dy
\end{equation}
and 
\begin{equation} \label{eps1}
\epsilon_1 = \int_{\R-[a,b] } V(y,T) f(y | x)dy
\end{equation}

$\epsilon_1$ corresponds to an integral truncation error. \\
Using (\ref{ThLegendreSeries}), it comes
\begin{equation}\label{ErrorAnv1v2}
V_1(x,0) = \sum_{k=0}^{\infty} A_kV_k = V_2(x,0) + \epsilon_2
\end{equation}
where $A_k$ is defined in (\ref{Ancoefficients}), $V_k$ in (\ref{VnExpression}) with
\begin{equation}\label{ErrorAnv2}
V_2(x,0) = \sum_{k=0}^{N-1} A_kV_k
\end{equation}
and 
\begin{equation}\label{eps2}
\epsilon_2 =  \sum_{k=N}^{+\infty} A_kV_k
\end{equation}
$\epsilon_2$ corresponds to the series truncation error.\\
By posing $C_m^k = \int_a^b P_k\left( \frac{2y-(a+b)}{b-a}\right) e^{i2\pi(\frac{ym}{b-a})} dy$, 
$A_k$ is written as
\begin{equation}\label{ErrorAnAk}
A_k = \frac{2k+1}{b-a}  \left[ \sum_{m=-\infty}^{+ \infty} B_m C_m^k \right]
\end{equation}
Using expression (\ref{ErrorAnv2}), we get 
\begin{equation}
V_2(x,0) = V_3(x,0) + \epsilon_3
\end{equation}
with 
\begin{equation}
V_3(x,0) = \frac{1}{b-a} \sum_{k=0}^{N-1} \sum_{m=-M}^M V_k (2k+1) B_m C_m^k
\end{equation}
and 
\begin{equation} \label{epsilon3}
\epsilon_3 = \sum_{k=0}^{N-1} \frac{V_k (2k+1)}{b-a} \sum_{m \, \in \Z -[-M,M]} B_m C_m^k
\end{equation}

$\epsilon_3$ represents another series truncation error.\\
Finally using (\ref{RelationBkBkTildeRk}), we have
\begin{equation}
V_3(x,0) = V_4(x,0) + \epsilon_4
\end{equation}
with 
\begin{equation}
V_4(x,t) = \frac{1}{b-a} \sum_{k=0}^{N-1} \sum_{m=-M}^M V_k (2k+1) \widetilde{B}_m C_m^k
\end{equation}
and 
\begin{equation}
\epsilon_4 = -\frac{1}{b-a} \sum_{k=0}^{N-1} \sum_{m=-M}^M V_k (2k+1) R_m C_m^k
\end{equation}

$\epsilon_4$ represents another integral truncation error.\\
To summarize we obtain
\begin{equation} \label{vv4eps1234}
V(x,0) = V_4(x,0) + \epsilon_1 + \epsilon_2 + \epsilon_3 + \epsilon_4. 
\end{equation}
$V_4(x,0)$ can be complex. By taking the real part in (\ref{vv4eps1234}), it comes
\begin{equation}
V(x,0) = Re(V_4(x,0)) + \epsilon_1 + \epsilon_2 + Re(\epsilon_3) + Re(\epsilon_4)
\end{equation}
because $V(x,0)$, $\epsilon_1$ and $\epsilon_2$ are real by definition.\\

Secondly, the key to bound the errors lies in the decay rate of the generalized Fourier series coefficients. 
The convergence rate depends on the smoothness of the functions on the expansion interval.\\ 

We summarize in the next theorem taken from \cite{WangXiang12}, 
the decay rates of the coefficients in the
Legendre series expansion and the error bounds of the truncated Legendre
series in the uniform norm. 

let $\| . \|_T$ be the Chebyshev-weighted seminorm defined by
$$\| u \|_T = \int_{-1}^1 \frac{\mid u'(x)\mid}{\sqrt{1-x^2}} dx $$,
$E_{\rho}$ denotes the Bernstein ellipse in the complex plane
$$E_{\rho} = \{ z \in \C \vert z = \frac{1}{2} (u + u^{-1}), u = \rho e^{i \theta}, -\pi \leq \theta \leq \pi  \} $$
and $$f_n(x) = \sum_{j=0}^n a_j P_j(x)$$
 the truncated Legendre series expansions of $f(x)$.

\begin{theorem}\label{WangXiangTheorem}
If $f$, $f'$,...,$f^{k}$  are absolutely continuous on $[-1,1]$ and 
$\| f^{(k)} \|_T = F_k < \infty$ for some $k \geq 1$ ($H_{abs}(k)$), then for each $n > k+1$, 
\begin{equation} \label{BoundAnAbsoluteCont}
\mid a_n \mid \leq \frac{F_k}{(n-\frac{1}{2}) (n-\frac{3}{2})...(n-\frac{2k-1}{2})} 
\sqrt{\frac{\pi}{2(n-k-1)}}.
\end{equation}

If $f$ is analytic inside and on the Bernstein ellipse $E_{\rho}$ with foci $\pm 1$ and major
semiaxis and minor semiaxis summing to $\rho >1$ ($H_a(E_{\rho})$), then for each $n \geq 0 $,
\begin{equation}\label{BoundAnAnalyticBern}
\mid a_n \mid \leq \frac{(2n+1) \ell(E_{\rho})M}{\pi \rho^{n+1}(1-\rho^{-2})}
\end{equation}
where $M = \max_{z \in E_{\rho}} \mid f(z) \mid$ and $\ell(E_{\rho})$ denotes the length of the circumference of $E_{\rho}$.
\end{theorem}

\subsection{ Bound for $\epsilon_2$ }

$A_k$ and $V_k$ correspond respectively to the Legendre series coefficients of the $f(x)$ and of the payoff functions. The density function is generally smoother than the payoff function in finance and we expect the coefficient $A_k$ to decay faster than $V_k$. The following proposition makes it precise.
\begin{proposition}\label{PropositionEpsilon2}
Let's assume $\int_a^b V^2(y,T)dy < +\infty$ and define 

\begin{equation}
g(y) \equiv f \left( \frac{b-a}{2}y + \frac{a+b}{2} \right).
\end{equation}

There are two cases:
\begin{enumerate}

\item Under $H_{abs}(k)$ with $k > 1$ for $g$, we have

\begin{equation}
\mid \epsilon_2 \mid \leq \frac{G_k}{(k-1)(N-\frac{1}{2})(N-\frac{3}{2})...(N-\frac{2k-3}{2})} 
		\sqrt{\frac{ \pi }{2(N-k)}}
\end{equation}
where $\| g^{(k)} \|_T = G_k < \infty$.

\item $g$ analytic on $[-1,1]$. Then we get
\begin{equation}
\mid \epsilon_2 \mid \leq \frac{(2N \rho +3\rho - 2N-1)\ell(E_{\rho})M }
{\pi \rho^{N+1}(\rho-1)^2(1-\rho^{-2}) }
\end{equation} 
where $\tilde{g}$ is the analytic continuation of $g$ on and within $E_{\rho}$ with $\rho > 1$,  $M \equiv \max_{z \in E_{\rho}} \mid \tilde{g}(z) \mid$ and $\ell(E_{\rho})$ denotes the length of the circumference of $E_{\rho}$.
\end{enumerate}

\end{proposition}

\begin{proof}

We have $V_k \to 0$ as $k \to \infty$. Indeed
\begin{align}
\mid V_k \mid &=\mid \int_{a}^{b}V(y,T) P_k \left( \frac{2y-(a+b)}{b-a} \right)dy  \mid \\
& \leq \| v \|_{L^2[a,b]} \sqrt{\frac{b-a}{2k+1}}.
\end{align}
where we have used Cauchy-Schwartz inequality and  $\| P_k \| = \sqrt{ \frac{2}{2k+1} }$.\\
So for $N >> \infty$, we write
\begin{equation}
\mid \epsilon_2 \mid \leq \sum_{k=N}^{\infty} \mid A_k.V_k\mid \leq \sum_{k=N}^{\infty} \mid A_k \mid  \equiv E(N)
\end{equation}  
Error $\epsilon_2$ is thus dominated by the Legendre series truncation of $f(x)$. \\
To bring the analysis in the interval $[-1,1]$, we perform a change of variable 
$y = \frac{2x-(a+b)}{b-a}$ and define $g(y) = f(\frac{b-a}{2}y + \frac{a+b}{2})$.\\
We consider 2 cases:
\begin{itemize}

\item Under $H_{abs}(k)$ with $k \geq 1$ for $g$, using (\ref{BoundAnAbsoluteCont}) and following the arguments in \cite{WangXiang12}, it comes
\begin{equation}
E(N) \leq \frac{G_k}{(k-1)(N-\frac{1}{2})(N-\frac{3}{2})...(N-\frac{2k-3}{2})} 
		\sqrt{\frac{ \pi }{2(N-k)}}
\end{equation} 

\item $g$ analytic on $[-1,1]$. \\
By the the theory of analytic continuation, there always exists a Bernstein ellipse 
$E_{\rho}$ with $\rho > 1$ such that $\tilde{g}$, the continuation of $g$, is analytic on and within $E_{\rho}$. Using (\ref{BoundAnAnalyticBern}) and following the arguments in \cite{WangXiang12}, we derive
\begin{equation}
E(N) \leq \frac{(2N \rho +3\rho - 2N-1)\ell(E_{\rho})M }
{\pi \rho^{N+1}(\rho-1)^2(1-\rho^{-2}) }
\end{equation} 

\end{itemize}
\end{proof}
\subsection{ Bound for $\epsilon_3$ }

\begin{proposition}\label{proposition:epsilon3}
If
\begin{enumerate}
\item 
\begin{equation} \label{Erroranalysisboundary}
f(b) = f(a), f^{(1)}(b) = f^{(1)}(a), ...,   f^{(l-1)}(b) = f^{(l-1)}(a) 
\end{equation}


\item $f^{(l)}(x)$ is integrable 

\end{enumerate}
Then
\begin{equation}
\epsilon_3 = \BTau(\frac{C_{N-1}}{M^l}) \hspace{0.5cm} for \mid M \mid >> 1
\end{equation}
with $C_{N-1} \equiv \sum_{k=0}^{N-1} \frac{\mid V_k \mid (2k+1)}{b-a}$.\\
In particular if the function $f$ is differentiable to all orders and $(1)$ is satisfied for any $l$, then $\epsilon_3$ decreases faster than $\frac{1}{\mid M \mid^l}$ for any finite power of $l$. This is the exponential convergence property. 
\end{proposition}

\begin{proof}
Let's fix $k \in [0,N-1]$. For $\mid M \mid >>1 $,
\begin{equation}
\mid \sum_{m \, \in \Z -[-M,M]} B_m C_m^k \mid \leq \sum_{m \, \in \Z -[-M,M]} \mid B_m \mid \mid C_m^k \mid 
\end{equation} 
By applying theorem 4 p.42 in \cite{Boyd00}, we get 
\begin{equation}\label{Bmlastbound}
\mid B_m  \mid \leq \frac{C_1}{\mid m \mid^l}
\end{equation} 
for a constant $C_1$ independent of $m$.\\
Using (\ref{LegendreFourierCoeff}), we have 
\begin{equation}\label{Cmkfirstbound}
\mid C_m^k  \mid \leq \frac{C_2}{ \sqrt{ \mid m \mid} } \mid  J_{k+\frac{1}{2}}(\pi m) \mid  
\end{equation}
for  a constant $C_2$ independent of $m, k$.\\
Applying the property $ n \in \Z, \, J_{\nu}(ze^{n \pi i}) = e^{n \nu \pi i} J_{\nu}(z)$ to $n=1$, we get
$\mid J_{\nu}(-z) \mid = \mid J_{\nu}(z)\mid$. Using the asymptotic result for $x \in \R, \, x \to +\infty$ (theorem 2.13 in \cite{MartinK})
\begin{equation}
J_{\nu}(x) \backsim \sqrt{\frac{2}{\pi x}} cos(x - \frac{\pi}{4} - \frac{\nu \pi}{2}),
\end{equation}

The expression in (\ref{Cmkfirstbound}) becomes, for $m >> 1$
\begin{equation}\label{Cmklastbound}
\mid C_m^k  \mid \leq \frac{C_3}{ \mid m \mid }.   
\end{equation}
with a constant $C_3$ independent of $m, k$.\\
So with (\ref{Bmlastbound}) and (\ref{Cmklastbound}), it comes
\begin{equation}
\mid B_m \mid C_m^k \mid \leq  \frac{C_4}{ \mid m \mid^{l+1} } 
\end{equation} 
for  a constant $C_4$ independent of $m, k$.\\
The series truncation error below (see \cite{BenderOrzag78} for proof) behaves like, for $M >> 1$,
\begin{equation}\label{SeriesTruncationAlg}
\sum_{m=M+1}^{\infty} \frac{1}{m^n}  \backsim \frac{1}{(n-1)M^{n-1}}.
\end{equation}
Applying (\ref{SeriesTruncationAlg}) with $n=l+1$, we finally obtain, for $M >>1$,
\begin{equation}
\sum_{m \, \in \Z -[-M,M]} \mid B_m \mid \mid C_m^k \mid  = \BTau \big(\frac{1}{M^l} \big).
\end{equation}
and can deduce directly $ \epsilon_3 = \BTau(\frac{C_{N-1}}{M^l})$.
\end{proof}

\remark{In practice, we believe the condition (\ref{Erroranalysisboundary}) in proposition \ref{proposition:epsilon3} is {\it{nearly}} satisfied if the boundary points $a$ and $b$ are chosen appropriately. Indeed,  $f(x)$, being a probability density function, converges to $0$ when $\mid x \mid \to \infty$.  Let's consider the benchmark and highly tractable Black Scholes model where the gaussian density function, $f_{m,\sigma}(x)$ with mean $m$ and standard deviation $\sigma$, and its  derivatives are known analytically and given by
\begin{equation}\label{GaussianDerivatives}
f^{(n)}_{m, \sigma}(x) = \frac{(-1)^n H_n(\frac{x-m}{\sigma}) f_{m, \sigma}(x)}{\sigma^n}
\end{equation}
with $H_n$ the {\it{Hermite polynomials}} defined by $H_n(x) = (-1)^n e^{\frac{x^2}{2}} \frac{d^n}{dx^n} e^{-\frac{x^2}{2}}$.\\
Figure \ref{figure:GaussianDerivatives} shows the graph of $f^{(n)}_{m, \sigma}$ for various values 
$n$. With $a = -1.7813$ and $b = 1.7188$, we observe clearly that condition (\ref{Erroranalysisboundary})
is {\it{closely}} satisfied. We will give insight in the choice of [a, b] in Section \ref{subsection:Truncation Range}. 

\begin{figure}[htbp]
  \centering
  \includegraphics[width=1.2 \textwidth, trim={1cm 7.5cm 1cm 6.5cm}, clip]{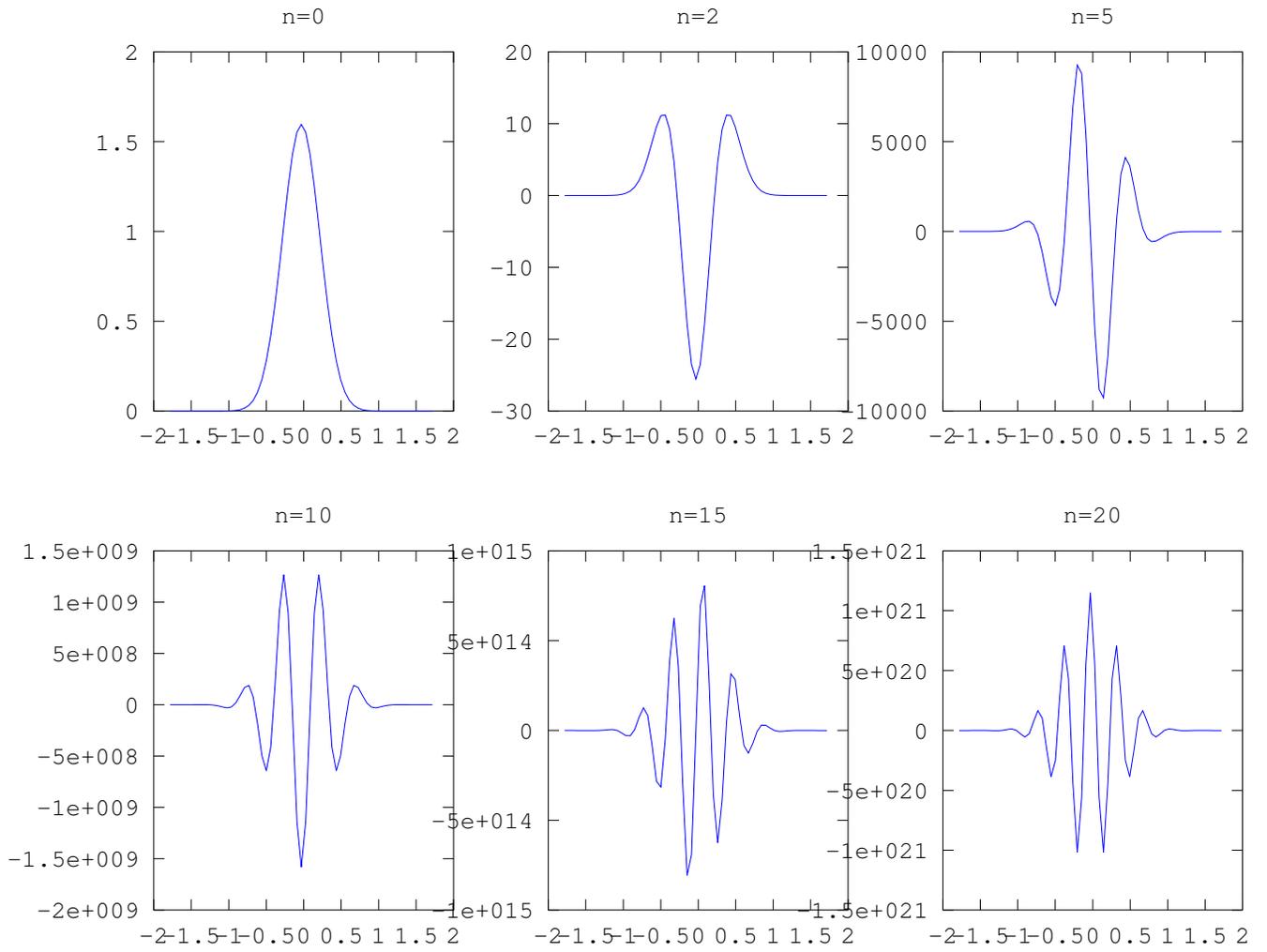}
  \caption{Various derivatives of the Gaussian density function in Black Scholes model (see section \ref{subsection:BSmodel}). The parameters are $S_0 = 1, r = 0, T = 1, \sigma = 0.25 $ and 
  $a = -1.7813$, $b = 1.7188$ with $L = 7$ for the truncation range (\ref{TruncationRange}).}
  \label{figure:GaussianDerivatives}
\end{figure}

\subsection{ Bound for $\epsilon_1$ and $\epsilon_4$ }
$\epsilon_1$ is simply bounded as $|\epsilon_1| \leq \int_{\R-[a,b] } |V(y,T)| \tilde{f}(y | x)dy$ and is small as soon as $\tilde{f}(y)$  decays to 0 faster than $V(y,T)$ in the tail.\\
$\epsilon_4$ is essentially bounded by the integral truncation of the density function as stated in the following proposition. 

\begin{proposition}\label{PropositionEpsilon4}
\begin{equation}
\mid \epsilon_4 \mid  \leq C_{N,M}. \epsilon
\end{equation}

where $C_{N,M} \equiv \frac{ \sum_{k=0}^{N-1} \sum_{m=-M}^M \mid V_k (2k+1) C_m^k \mid }{b-a}$ 
and $\epsilon \equiv \frac{1}{b-a} \left[ \tilde{F}(a) + 1 - \tilde{F}(b) \right]$ with $\tilde{F}(x)$  the cumulative distribution function of $\tilde{f}(x)$.
\end{proposition}

\begin{proof}
$R_k$, being defined in (\ref{Rkdefinition}), can be bounded as
\begin{equation}
\mid R_k \mid \leq \frac{1}{b-a} \left[  \int_{-\infty}^{a} \tilde{f}(x) dx + \int_{b}^{ +\infty } \tilde{f}(x) dx \right] = \frac{1}{b-a} \left[ \tilde{F}(a) + 1 - \tilde{F}(b) \right] = \epsilon
\end{equation}
It comes
\begin{equation}
\mid \epsilon_4 \mid  \leq \frac{ \sum_{k=0}^{N-1} \sum_{m=-M}^M \mid V_k (2k+1) C_m^k \mid }{b-a} \epsilon = C_{N,M} \epsilon.
\end{equation}
\end{proof}

\remark{ \label{regularitypayoff}

\begin{itemize}

\item Our error analysis relies on the smoothness of the density function and not on the regularity of the payoff function. 
We just require the payoff function some integrability conditions on bounded interval and that the density function decays faster to $0$ at infinity. 
This is particularly relevant in quantitative finance. Indeed, the density functions of asset prices tend to be smoother. 
And the payoffs of some contracts are discontinuous as for the digital option (\ref{digital}) or have a kink at the strike level as for call and put options (\ref{callput}).

\item Some well-established option pricing methods depend on the regularity of the payoff function. For example, in the Carr-Madan approach \cite{CarrMadan99} and its variants, the Fourier transform of a version of valuation formula (\ref{riskneutralvaluation}) is taken with respect to the log-strike price. Damping of the payoff is then necessary as, for example, a call option is not $L^1$-integrable with respect to the logarithm of the strike price. The method’s accuracy depends on the correct value of the damping parameter.
Or in the widely used Monte carlo method with discretisation of SDEs (e.g Euler or Milstein schemes), the smoothness of the payoff function impact directly the order of convergence of the approximation schemes (see \cite{Glasserman03} section 6 or \cite{KloedenPlaten92} section 14.5).  


\end{itemize}
}

\section{Numerical experiments} \label{sec:Numerical experiments}

In this section, we perform a variety of numerical tests to illustrate 
the robustness and accuracy of the new computational pricing method using Legendre polynomial. 
The payoff functions in finance can be continuous or discontinuous. Here the European call options and the European digital call options are considered. 
It allows to show that the convergence of the pricing method using Legendre polynomial does not depend on the continuity of the payoff (see discussion in section \ref{subsec:optionpricing} and remark \ref{regularitypayoff}). We cover a representative class of models widely studied and used in quantitative finance: 
\begin{itemize}
\item Black Scholes  Model;
\item Merton Jump Diffusion Models and Kou Jump Diffusion Models;
\item Heston Stochastic Volatility Model.
\end{itemize}

They represent different schools of thoughts for the modelling of asset prices as random processes. 
In their seminal work in \cite{BlackScholes73}, Black and Scholes modelize the asset prices as a {\it{geometric Brownian motion}} i.e asset prices with continuous paths and a constant volatility. It leads to the famous Black-Scholes formula which gives a theoretical estimate of the price of European-style options. It is perhaps the world's most well-known options pricing model and usually used as a benchmark model by the quantitative finance community. 

However, one of the main shortcoming of Black and Scholes model is to assume the underlying volatility is constant over the life of the derivative, and unaffected by the changes in the price level of the underlying security.  It cannot explain 
long-observed features of the implied volatility surface such as volatility smile and skew, which indicate that implied volatility does tend to vary with respect to strike price and expiry. By assuming that the volatility of the underlying price is a stochastic process rather than a constant, it becomes possible to model derivatives more accurately (see e.g \cite{Gatheral06} and \cite{Wilmott06}). And Heston model is one of the most popular stochastic volatility models for derivatives pricing.

An another school of modelling asset prices consists to introduce jumps as a way to explain why the skew is so steep for very short expirations and why the very short-dated term structure of skew is inconsistent with any stochastic volatility model. 
Or the strongest argument for using discontinuous models is simply the presence of jumps in observed prices (see figure 1 in \cite{TankovVoltchkova09} and \cite{ContTankov04} or \cite{Gatheral06}). Merton jump diffusion model and Kou jump diffusion model are among the most popular jumps models used in quantitative finance. 

In the equity and exchange rates (FX) derivatives market, liquid options like European call or put contracts are quoted for different maturities or tenors with various strikes representing the moneyness. FX markets are particularly liquid at benchmark tenors, such as 1 month (M), 2M, 3M, 6M, 1 year (Y), 2Y and possibly longer dated options \cite{Clark11}. For liquid equity index like Eurostoxx 50 or Nikkei 225, we can observed quotes for maturities from 1 month up to 4 and 5 years \cite{EDS08}. With this in mind, for the tests to be comprehensive, we consider short, standard and long maturities (0.1, 1, 3, 10 years) and in/at/out of the money options.

\subsection{Truncation Range} \label{subsection:Truncation Range}

For practical usage, it is important to determine appropriately and as systematically as possible the range $[a,b]$ to minimise the integral truncation errors $\epsilon_1$ and $\epsilon_4$. Being given the characteristic function of $X = \log( \frac{S_T}{K} )$, we can compute its cumulants $c_n$, defined in (\ref{CumulantsCharacteristics}), and uses the following formula proposed in \cite{FangOosterlee08}:
\begin{equation} \label{TruncationRange}
[a,b] := \left[ c_1 - L \sqrt{ c_2 + \sqrt{c_4} }, c_1 + L \sqrt{ c_2 + \sqrt{c_4} } \right]
\end{equation}

The cumulants for each model are given in appendix B. 
As shown in the error analysis section, the accuracy of the Legendre polynomial pricing method is affected by the choice of the interval $[a, b]$.
Some experience is helpful when choosing the correct truncation range. The value for $L$ is taken to be in $[7,12]$ and will be made explicit for each model in the tests. 

Expression (\ref{TruncationRange}) uses $c_n$ up to degree 4. 
Similar range formula involving the first two moments of $X$ is implemented in \cite{HurnLindsayClelland13}.
In general, using high order cumulants captures better the tail behaviour of the distribution. 

\subsection{Black Scholes  Model}\label{subsection:BSmodel}

For this Model, the SDE for the asset price $S_t$, under risk neutral measure, is given by
\begin{equation} \label{BSSDE}
\frac{dS_t}{S_t} = rdt + \sigma dW_t.
\end{equation}
with $W_t$ the Brownian motion, $r$ the risk free rate and $\sigma$ the volatility parameter.\\
The characteristic function $\varphi(u)$ of $\log(\frac{S_T}{K})$ is 
\begin{equation} \label{BSCharacteristicF}
\varphi(u) = \exp(\mu u i - \frac{1}{2} u^2 \sigma^2T)
\end{equation}
with $\mu = -\frac{1}{2}\sigma^2 T-\log(K)$.\\
The set of parameters is $S_0 = 1, r = 0, T = 10, \sigma = 0.25 $.
With some experiments, choosing $L$ around $7$ is appropriate for the truncation range (\ref{TruncationRange}) and this value is consistent with thosed used in \cite{HurnLindsayClelland13}.\\
The other details of the model are provided in section \ref{annex:BS}.\\
We examine a long maturity with $T=10$. Figure \ref{BSDensity} compares the true Gaussian density and the recovered density functions using respectively $N=M=12$ and $N=M=32$ for the truncation in formula (\ref{ThLegendreSeries}). With $N=M=12$, the approximating density captures the form of the true density although we observe some  slight negative values in the tail. With  $N=M=32$, it is indistinguishable from the true density function. \\
For the pricing, we consider a discontinous payoff with the digital option. As shown in 
Figure \ref{BSCvTests}, the error convergence of the method is exponential in terms of $N$ and $M$ respectively. Indeed, in Black Scholes model, the density function of $\log ( \frac{S_T}{K})$ is gaussian and so is infinitely differentiable with exponential decay to 0 for large $x$. Further, we observe the error convergence rate is basically the same for different strike prices.  

\begin{figure}[htbp]
  \centering
  \includegraphics[width=0.6 \textwidth, height=0.4 \textwidth, trim={1cm 6.5cm 1cm 6cm}, clip]{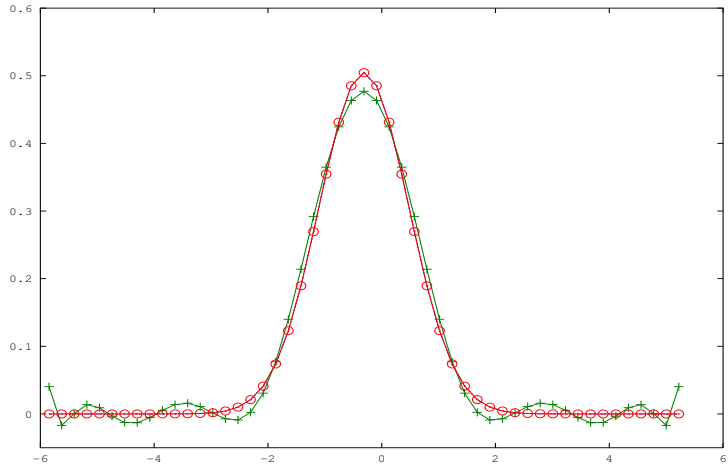}
  \caption{Comparison of the true Gaussian density (solide line) and 
  its approximation based on $N=M=12$ (solide line with '+' marker) and $N=M=32$ (solide line with 'o' marker) for maturity $T=10$.}
  \label{BSDensity}
\end{figure}

\begin{figure}[htbp]
  \centering
  \includegraphics[width=0.9 \textwidth, trim={1cm 7.5cm 1cm 8cm}, clip]{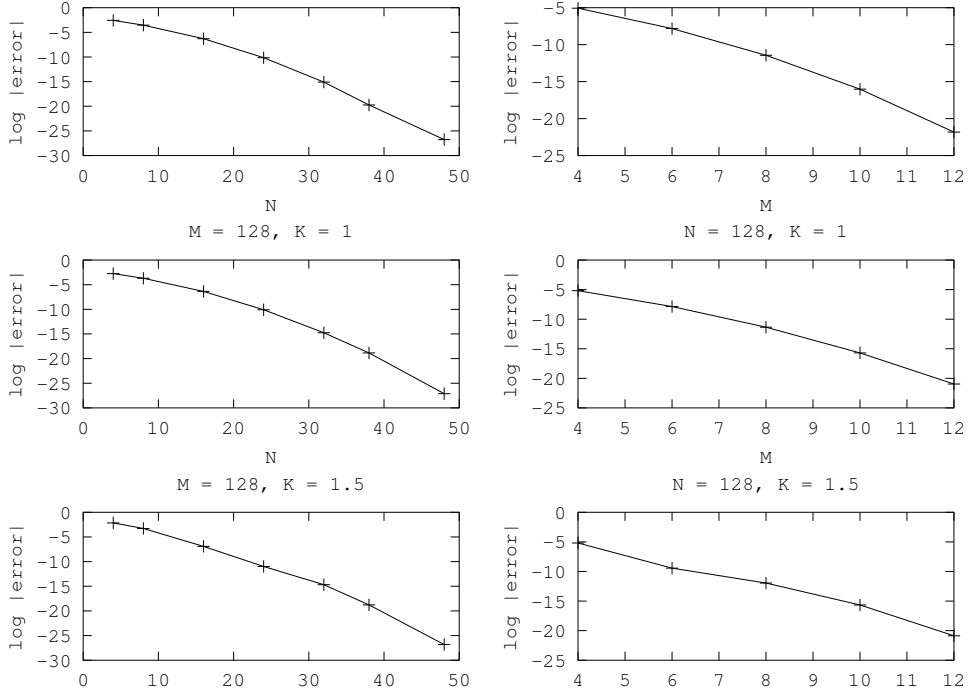}
  \caption{Error convergence for pricing European digital call option with $T = 10$ in Black Scholes model.}
  \label{BSCvTests}
\end{figure}

\subsection{Merton Jump Diffusion Model}
In this model \cite{Merton76, AndersenAndreasen99}, 
the SDE for the asset price $S_t$, under risk neutral measure, is written as
\begin{equation} \label{MertonSDE}
\frac{dS_t}{S_t-} = (r-\lambda.m(t))dt + \sigma dW_t + (J(t)-1)d\pi(t).
\end{equation}
where $W(t)$ is a Brownian motion, $\pi(t)$ a Poisson counting process with constant jump intensity 
$\lambda$ and r the deterministic risk-free interest rate.
$\{J(t)\}_{t \geq 0}$, representing the jump size, is a sequence of independent log normal random variables 
of the form  $J(t) = e^{\mu + \gamma N(t)}$ with $N(t)$ a standard gaussian random variable and $m \equiv E[J(t)-1]$. $\pi, \, W$ and $J$ are all assumed to be independent.  \\
The characteristic function $\varphi(u)$ of $\log(\frac{S_T}{K})$ is
\begin{equation} \label{MertonCharacteristicF}
\varphi(u) = \exp \left( iu\tilde{b}T - \frac{u^2 \sigma^2T}{2} + \lambda T(e^{iu \mu - \frac{\gamma^2 u^2}{2}} -1)  \right)
\end{equation}
where $\tilde{b} = b - \frac{\log(K)}{T}$ and $b = -\frac{1}{2} \sigma^2 - \lambda (e^{\mu + \frac{\gamma^2}{2}}-1) $.\\
The set of parameters is calibrated to market data from \cite{AndersenAndreasen99} with $r=0$ and maturity 3 years:
$S_0 = 1, T = 3, \, \sigma = 0.1765 $, $\lambda = 0.089$, 
$\mu = -0.8898$, $\gamma = 0.4505$. Some experience shows that $L = 10$ is appropriate for the truncation range (\ref{TruncationRange}). It corresponds also to the value recommended in \cite{FangOosterlee08}. 
 The other details of the model are provided in section \ref{annex:Merton}.\\
We study a standard maturity with $T=3$. Figure \ref{MertonDensity} compares the true density and the recovered density functions using respectively $N=M=50$ and $N=M=80$.
The true density is computed using formula (\ref{Mertondensity}) with a truncation in the infinite sum at $50$. First we observe the {\it{Merton}} density function, showing a sharp peak, is less smooth than the Gaussian density function. 
With $N=M=50$, the approximating density captures reasonably well the form of the true density although we can observe some slight negative values in low probability area. With  $N=M=80$, the difference between the true density and the approximating density is not discernible.  \\
We consider the digital option for pricing. As shown in 
Figure \ref{MertonCvTests}, the error convergence of the method is still exponential in $N$ and $M$ respectively. 
But the convergence is slower than in the Black Scholes model as expected in view of the sharp peak density function.  
And the error convergence rate is basically the same for different strike prices.  
\begin{figure}[htbp]
  \centering
  \includegraphics[width=0.8 \textwidth, height=0.5 \textwidth, trim={1cm 6.5cm 1cm 6cm}, clip]{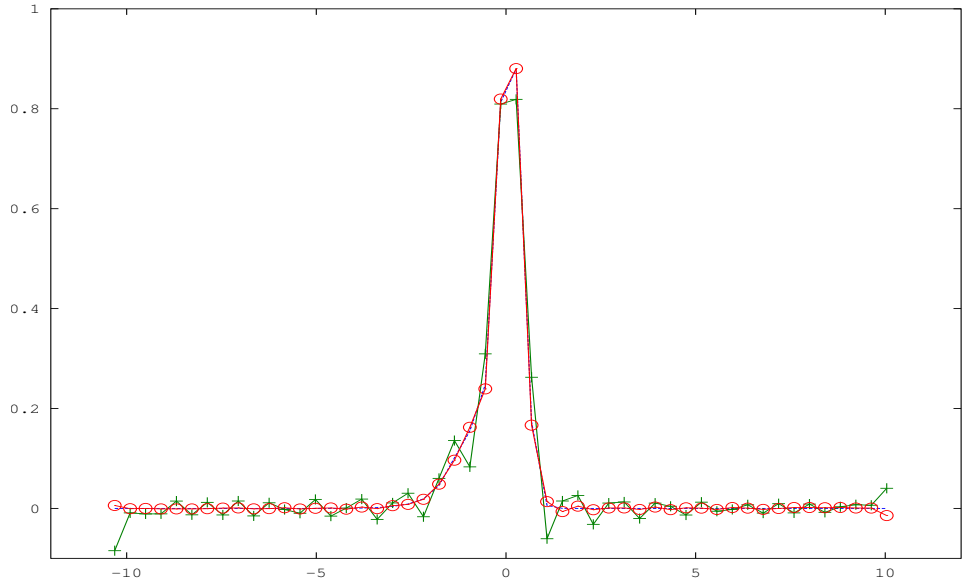}
  \caption{Comparison of the true density function,  (solide line) and 
  its approximation based on $N=M=50$ (solide line with '+' marker) and $N=M=84$ (solide line with 'o' marker) for maturity $T=3$ in Merton jump diffusion model.}
  \label{MertonDensity}
\end{figure}

\begin{figure}[htbp]
  \centering
  \includegraphics[width=0.85 \textwidth, trim={1cm 7.5cm 1cm 6cm}, clip]{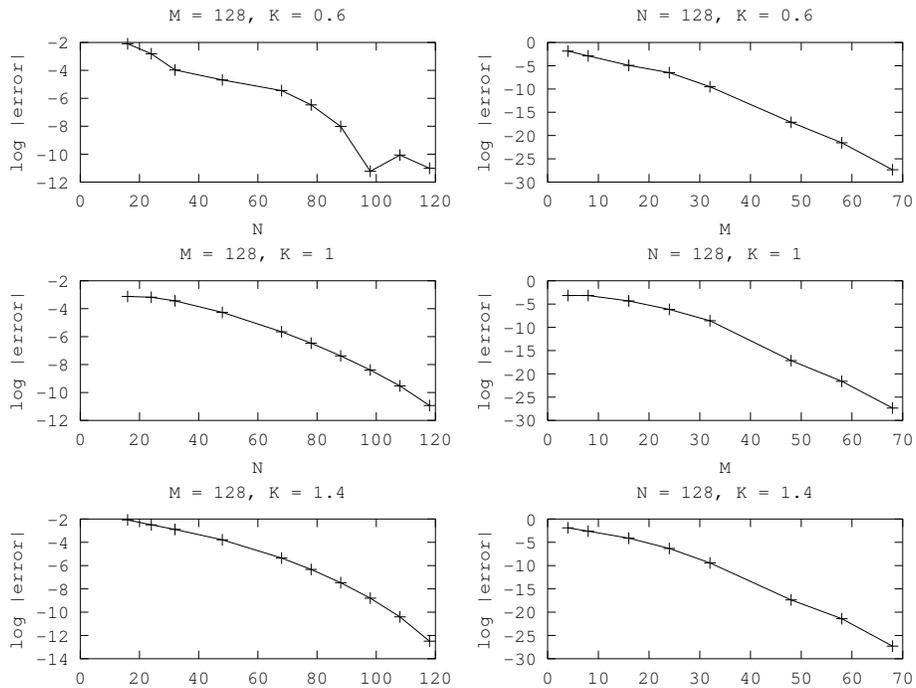}
  \caption{Error convergence for pricing European digital call option with $T = 3$ in Merton Jump Diffusion model.}
  \label{MertonCvTests}
\end{figure}

\subsection{Kou Jump Diffusion Model}

In this model \cite{Kou02}, the dynamic of the asset price, $S(t)$, under risk neutral probability, is 
\begin{equation} \label{KouSDE}
\frac{dS_t}{S_t-} = \mu dt + \sigma dW_t + d \left( \sum_{i=1}^{N(t)}(V_i-1) \right).
\end{equation}

with $W_t$ a standard Brownian motion, $N(t)$ a Poisson process with rate $\lambda$. 
$\{ V_i \}$ is a sequence of independent identically distributed (i.i.d.) nonnegative random variables
such that $Y = \log(V)$ has an asymmetric double exponential distribution with the density 

\begin{equation} 
f_{Y}(y) = p. \eta_1 e^{- \eta_1 y}1_{y \geq 0} + q. \eta_2 e^{ \eta_2 y}1_{y < 0}, \,\, \eta_1>1, 
\eta_2 > 0,
\end{equation}

$p,q \geq 0, \, p+q = 1$, representing the probabilities of upward and downward jumps and
$\mu = \lambda \left( \frac{p}{1- \eta_1} + \frac{1-p}{\eta_2+1} \right)$.\\
The characteristic function $\varphi(u)$ of $\log(\frac{S_T}{K})$ is given by

\begin{equation} \label{KouCharacteristicF}
\varphi(u) = \exp \left( iu\tilde{b}T - \frac{u^2 \sigma^2T}{2} + 
\lambda T iu \left(  \frac{p}{\eta_1-iu}-\frac{1-p}{\eta_2+iu} \right)  \right)
\end{equation}

where $\tilde{b} = b - \frac{\log(K)}{T}$ and $b = -\frac{1}{2} \sigma^2 - \lambda (e^{\mu + \frac{\gamma^2}{2}}-1) $. \\
The set of parameters is from \cite{Kou02} with $r=0$ and maturity 1 year:\\
$S_0 = 1, T = 1, \sigma = 0.16 $, $\lambda = 1$, 
$p = 0.4$ and $\eta_1 = 10, \, \eta_2 = 5$.
Some experience shows that $L = 10$ is appropriate for the truncation range (\ref{TruncationRange}). It corresponds also to the value recommended in \cite{FangOosterlee08}.  
The other details of the model are provided in section \ref{annex:Kou}.\\
Analytical formula for density function being not available,  Figure \ref{KouDensity} presents the recovered density functions for $T = 1$ and $3$ respectively. We observe a sharper-peaked density for $T=1$.\\
For the pricing, we study the European call option with a standard maturity $T=1$ y. 
As shown in Figure \ref{KouCvTests}, the error convergence of the method is exponential in $N$ and $M$ respectively. Further, the error convergence rate is basically the same for different strike prices. 

\begin{figure}[htbp]
  \centering
  \includegraphics[width=0.7 \textwidth, height=0.5 \textwidth, trim={1cm 6.6cm 1cm 6cm}, clip]{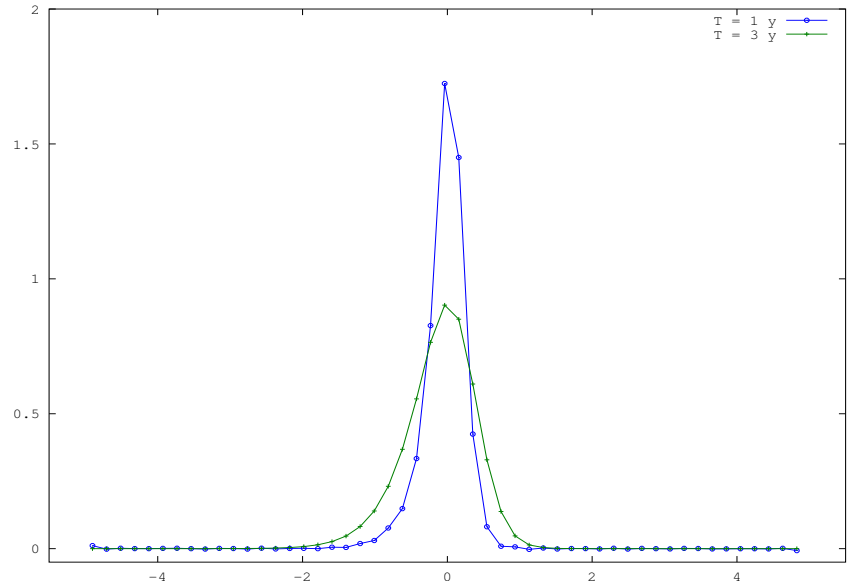}
  \caption{ Recovered density functions in Kou jump diffusion model for $T = 1$y (solide line with 'o' marker) and $T = 3$y (solide line with '+' marker) with $N=M=80$.}   
  \label{KouDensity}
\end{figure}

\begin{figure}[htbp]
  \centering
  \includegraphics[width=0.85 \textwidth, trim={1cm 7.5cm 1cm 6cm}, clip]{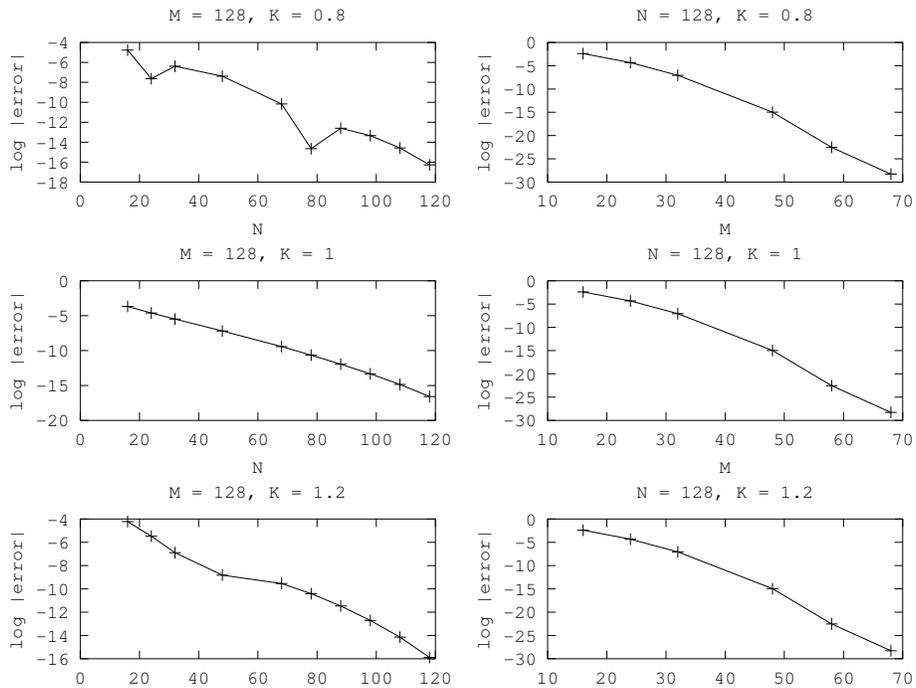}
  \caption{Error convergence for pricing European call option with $T = 1$ in Kou Jump Diffusion model.}
  \label{KouCvTests}
\end{figure}

\subsection{Heston Stochastic Volatility Model }

In this model \cite{Heston93} under risk neutral measure, the SDEs are given by
\begin{align}\label{HestonSDE}
d\tilde{x}_t & = -\frac{1}{2}u_tdt + \sqrt{u_t}dW_{1t}\\
du_t & = \lambda( \bar{u} - u_t)dt + \eta \sqrt{u_t}dW_{2t}
\end{align}

where $\tilde{x}_t$ denotes the log-asset price variable and $u_t$ the variance the asset price process. 
Parameters $\lambda \geq 0, \ \bar{u} \geq 0$ and $\eta \geq 0$ represent the speed of mean reversion, the mean level of variance and the volatility of volatility, respectively. Furthermore, the Brownian motions $W_{1t}$ and $W_{2t}$ are assumed to be 
correlated with correlation coefficient $\rho$. \\
The characteristic function $\varphi(x)$ of $\log(\frac{S_T}{K})$ can be represented by
\begin{equation} \label{HestonCharacteristicF}
\varphi(x) = e^{ -ix \log(K)  + \frac{u_0}{\eta^2} \left( \frac{1-e^{-DT}}{1-G\e^{-DT}} \right) 
(\lambda - i \rho \eta x -D) + \frac{\lambda \bar{u}}{\eta^2} \left[ T(\lambda - i \rho \eta x -D)
-2 \log \left( \frac{1-Ge^{-DT}}{1-G}  \right) \right]} 
\end{equation}

with $D = \sqrt{(\lambda - i \rho \eta x)^2 + (x^2 + ix) \eta^2}$ and 
$G = \frac{\lambda - i\rho \eta x -D}{\lambda - i\rho \eta x + D}$.\\
This characteristic function is uniquely specified, since we take $\sqrt{x+yi}$ such that its real part is nonnegative, and we restrict the complex logarithm to its principal branch.
In this case the resulting characteristic function is the correct one for all complex $\omega$ in the strip of analyticity of the characteristic function \cite{LordKahl10}. \\
The set of parameters is calibrated to market data from \cite{Crisostomo14} with $r=0$ and a short maturity $T = 0.1$:\\
$S_0 = 1, \, T = 0.1, \, \lambda = 0.9626, \, \bar{u} = 0.2957, \, \eta = 0.7544,
\rho = -0.2919, \, u_0 = 0.0983$. \\
Since the analytical formula for 
$c_4$ is involved, instead of  
(\ref{TruncationRange}), as recommended in \cite{FangOosterlee08},  we use the following truncation range:

\begin{equation} \label{TruncationRangeHeston}
[a,b] := \left[ c_1 - 12 \sqrt{ |c_2| }, c_1 + 12 \sqrt{ |c_2| } \right]
\end{equation}

Cumulant $c_2$ may become negative for sets of Heston parameters that do not satisfy
the Feller condition, i.e, $2 \bar{u} \lambda > \eta^2$. We therefore use the absolute 
value of $c_2$. The formulas for cumulants are reported in section (\ref{annex:Heston}).\\
Analytical formula for density function being not available,  Figure \ref{HestonDensity} provides an illustration 
for the recovered density functions with $T = 0.1$ and $T = 1$ respectively. For $T=0.1$, the density is much more peaked.  \\
We examine the European call option with a short maturity $T=0.1$ year for pricing. Although the sharp peaked density, the error convergence of the method is still exponential in $N$ and $M$ respectively, as shown in Figure \ref{HestonCvTests}. Moreover, the error convergence rate is basically the same for different strike prices. 

\begin{figure}[htbp]
  \centering
  \includegraphics[width=0.73 \textwidth, height=0.5 \textwidth, trim={1cm 6.5cm 1cm 6cm}, clip]{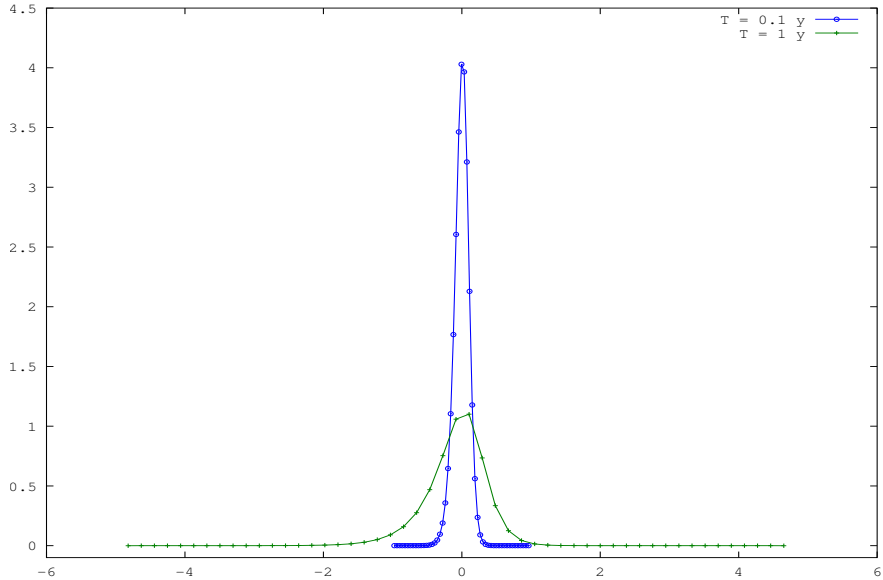}
  \caption{ Recovered density functions in Heston stochastic volatility model for $T = 0.1$ y (solide line with 'o' marker) and $T = 1$ y (solide line with '+' marker) with $N=M=80$.}   
  \label{HestonDensity}
\end{figure}

\begin{figure}[htbp]
  \centering
  \includegraphics[width=0.85 \textwidth, trim={1cm 7.3cm 1cm 7.2cm}, clip]{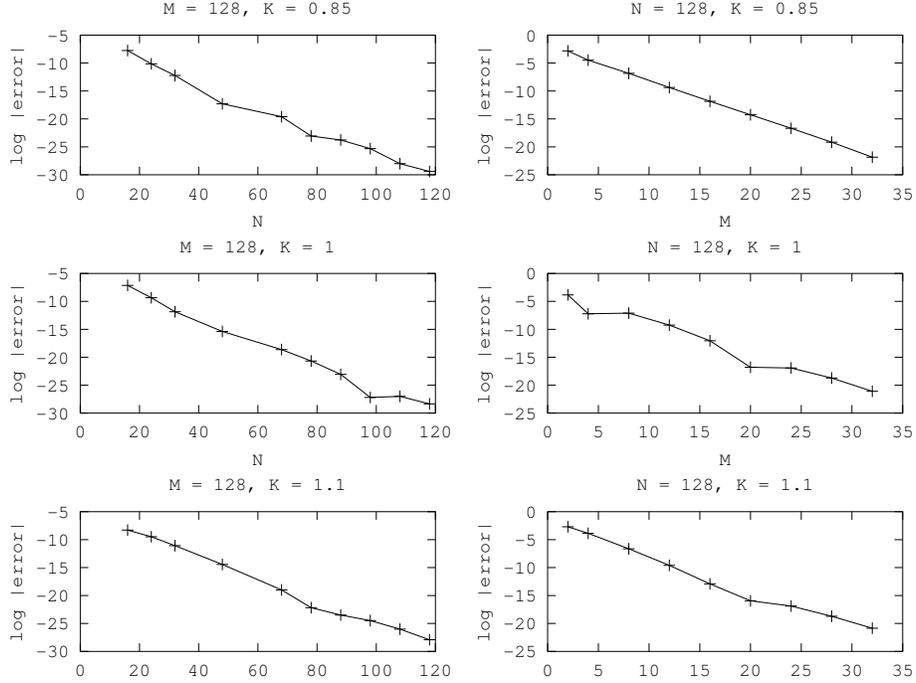}
  \caption{Error convergence for pricing European call option with short maturity $T = 0.1$ in Heston stochastic volatility model.}
  \label{HestonCvTests}
\end{figure}

\section{Conclusion and discussions}





In this paper, we have introduced a method for pricing European-style options combining Fourier series 
and generalized Fourier series with Legendre polynomials.
It can be used as long as a characteristic function for the underlying price process is available. 
It consists to expand the density function as Legendre series and observe that the coefficients can be accurately retrieved from the characteristic functions. This representation of the density function is then uses for the risk-neutral valuation and approximation formulas are derived. These formulas involve the expression (\ref{IntLegendreExp}). However its direct implementation, using formulation (\ref{PropIntLegendreExp}), gives rise of rapid accumulations of round-off errors for large values of $n$. We rewrite these quantities as solution of second-order difference equations and compute them with machine precision using the stable {\it{Olver}} algorithm. 
Also derivation of the pricing method has been accompanied by an error analysis.   
Errors bounds have been derived and the study relies more on smoothness properties which are not provided by the payoff functions, but rather by the density function of the underlying stochastic models. 
This is particularly relevant in quantitative finance for option pricing where the payoffs of the contract are generally not smooth functions. In our numerical experiments, we chose a class of models widely used in quantitative finance. The payoff covered are continuous (call option) and discontinuous (digital call option). 
The tests considered, with various strike prices and maturities, show exponential convergence rate.\\

We suggest a couple of interesting avenues of research: 

\begin{itemize}

\item Here, we have used Olver's algorithm for the computations of the integrals involving Legendre polynomials and exponential functions (\ref{IntLegendreExp}). Indeed, Olver's method consists to replace the original problem by an equivalent boundary value problem, which is solved by Gaussian elimination without pivoting.  Extensions or reformulations of Olver's method have been made. For example Van der Cruyssen \cite{Cruyssen79} 
have shown that if the algebraic equations arising from the use of Olver's method are solved using an LU decomposition method then the total amount of work required is almost halved. See Lozier \cite{Lozier80} for a more detailed discussion. It would be interesting to reconsider and adapt or extend existing algorithms  in order to reduce the amount of computational effort.

\item Accurately valuing financial claims plays a key role in financial modelling, but the risk management of these derivative instruments is at least as important (see e.g \cite{Wilmott06} or \cite{AvellanedaLaurence99}). To undertake this function, we need to compute the {\it{Greeks}} defined as the sensitivity of the price of derivatives to a change in underlying parameters on which the value of an instrument  is dependent. Series expansions for the sensitivity factors, e.g $\Delta = \frac{\partial V}{\partial S_0}$, 
$ \Gamma = \frac{\partial^2 V}{\partial S_0^2}$ or $\nu = \frac{\partial V}{\partial \sigma}$ are let for future research.

\item In this manuscript, we have focused on the pricing of European-style options, which are instrumental and  the building blocks for constructing more complex option products. Extending Legendre polynomials pricing method to cover  more exotic contracts like forward start options,  quanto options,  spread options or options with early-exercise features (see e.g \cite{BompisHok14}, \cite{Pelsser00}, \cite{Haug07}) are exciting area of research. 

\item The calibration, which consists to determine the parameters of a parametric model, is an instrumental preliminary step for option pricing and hedging. Usually, it corresponds to find parameters that make the models consistent with market quotes (e.g a set of European call or put prices for various strikes and maturities) and the formulas derived in proposition \ref{PropositionPricingFormulas} can be used. This is formulated as a minimisation of some loss functions (e.g the squared difference between the quoted and model prices) and commonly leads to a non convex optimisation problem. Standard procedures based on the derivatives of the loss function (e.g quasi-Newton Broyden-Fletcher-Goldfarb-Shanno (BFGS) algorithm) may be not appropriate. Indeed different starting points can lead to quite different solutions, which can have a significant impact on option pricing and sensitivity factors \cite{AzencottGadhyanGlowinski15}. Similarly, different calibration criteria lead to different results \cite{DetlefsenHardle07}. In \cite{GiliSchumann11}, the authors suggested to use heuristic techniques, differential evolution and particle swarm optimisation, which seem to bring some promising results. Exploring recent literature about real life applications of contemporary numerical optimisation and classification techniques in different fields such as \cite{TaorminaChau15, ZhangChau09, DimensionReductionZhangChau09, MayDandyMaier11} is part of our future research.

\end{itemize}
   
\section{Annexe}

\subsection{Appendix A}

We propose in the following proposition some analytical formulas for the computation of expression 
\begin{equation}\label{IntLegendreExp}
\int_{\alpha}^1 P_n(t)e^{\beta t}dt.
\end{equation}

\begin{proposition}\label{PropIntLegendreExp}

for $\beta \neq 0$, 
\begin{align}
\int_{\alpha}^1 P_n(t) e^{\beta t}dt =& \frac{1}{2^n} \sum_{k=0}^{\lfloor \frac{n}{2} \rfloor} (-1)^k C_n^k C_{2n-2k}^n \left[ IEP(\beta, n-2k,1)-IEP(\beta, n-2k,\alpha) \right]  \label{IntLegendreExp1} 
\end{align}
with 
\begin{equation}
IEP(\beta,n,t):=e^{\beta t} \sum_{i=0}^n t^i \frac{n!}{\beta^{n+1-i} i!}(-1)^{n-i}.
\end{equation}
\end{proposition}

\begin{proof}
For (\ref{IntLegendreExp1}), by using successively integration by parts formula, we show easily
\begin{equation}
\int  t^n e^{\beta t} dt = e^{\beta t} \sum_{i=0}^n t^i \frac{n!}{\beta^{n+1-i} i!}(-1)^{n-i}.
\end{equation}

We obtain (\ref{IntLegendreExp1}) by using the formula (\ref{LegendrePowerSeries}) for Legendre polynomial.  
\end{proof}

The computation of (\ref{IntLegendreExp}) with (\ref{IntLegendreExp1}) for small values of $n$ does not pose any problems. When $n >> 1$ accuracy and stability issues arise because of serious substractive cancellations in the summation (\ref{IntLegendreExp1}). In section (\ref{subsection:AlternateComputational}), we propose to use Olver algorithm to implement these terms in a stable way with machine precision. 

\subsection{Appendix B}

Here, we provide some analytical formulas for the  cumulants of $\log(\frac{S_T}{K})$, the density function of $\log(S_T)$ and the option pricing in the class of models discussed in section \ref{sec:Numerical experiments}.\\
Let $X$ be a random variable and $\Phi_X$ its characteristic function. We can define an unique continuous function $\Psi_X$ in a neighbourhood of zero 
such that
\begin{equation}
\Psi_X(0) = 0 \,\,\,\, \textrm{and} \,\,\,\, \Phi_X(z) = exp[\Psi_X(z)].
\end{equation}
The function $\Psi_X$ is called the {\it{cumulant generating function}}.
The {\it{cumulants}} of $X$ are defined by

\begin{equation} \label{CumulantsCharacteristics}
c_n(X) = \frac{1}{i^n} \frac{\partial^n \Psi_X}{\partial u^n}(0).
\end{equation}
(see \cite{ContTankov04} for details). We present the cumulants $c_1$, $c_2$ and $c_4$, needed to determine the truncation range in (\ref{TruncationRange}). 

\subsubsection{Black Scholes model} \label{annex:BS}

With $r=0$ in (\ref{BSSDE}), we have
\begin{align}
c_1& =  \log \left( \frac{S_0}{K} \right) - \frac{1}{2}\sigma^2 T\\
c_2& = \sigma^2 T\\
c_4& = 0 \\
\log(S_T) & \sim N( \log(S_0) - \frac{1}{2}\sigma^2 T, \sigma \sqrt{T})\\
V(x,t) & = N \left( \frac{\log(\frac{S_0}{K}) - \frac{1}{2} \sigma^2 T}{ \sigma \sqrt{T}} \right)
\end{align}
where $V(x,t)$ is the analytical digital call price with strike $K$.

\subsubsection{Merton Jump Diffusion Model} \label{annex:Merton}

With $r=0$ in (\ref{MertonSDE}), we have

\begin{align}
c_1& =  T (\tilde{b} + \lambda \mu)\\
c_2& = T( \sigma^2 + \lambda(\mu^2 + \gamma^2)) \\
c_4& = T \lambda (3 \gamma^4 + 6 \mu^2 \gamma^2 + \mu^4) \\
f_{X_T}(x)& = e^{- \lambda T} \sum_{k=0}^{\infty} \frac{ (\lambda T)^k}{k!} 
\frac{1}{\sqrt{ 2 \pi (\sigma^2T + k \gamma^2)}} e^{-\frac{1}{2} \frac{(x - \tilde{b}T - k \mu)^2}{\sigma^2 T + k \gamma^2}}  \label{Mertondensity} \\
V(x,t) & = e^{- \lambda T} \sum_{k=0}^{\infty} \frac{ (\lambda T)^k}{k!} N \left( \frac{ \log(\frac{S_0}{K}) + bT + k \mu }{\sqrt{\sigma^2 T + k \gamma^2}} \right)
\end{align}

with $b = -\frac{1}{2} \sigma^2 - \lambda(e^{\mu + \frac{\gamma^2}{2}} - 1) $, $\tilde{b} = b - \frac{\log(K)}{T}$, $f_{X_T}(x)$ the probability density function 
of $X_T \equiv \log(\frac{S_T}{K})$ and $V(x,t)$ the analytical digital call price with strike $K$.

\subsubsection{Kou Jump Diffusion Model} \label{annex:Kou}

With (\ref{KouSDE}), we have

\begin{align}
c_1& =  T \left(\tilde{b} + \lambda \left( \frac{p}{\eta_1}-\frac{1-p}{\eta_2}  \right) \right)\\
c_2& = T \left( \sigma^2 + 2\lambda \left( \frac{p}{\eta_1^2}+\frac{1-p}{\eta_2^2}  \right)\right) \\
c_4& = 24T \lambda \left( \frac{p}{\eta_1^4}+\frac{1-p}{\eta_2^4}  \right) 
\end{align}

The pricing formula for a European call option is involved and can be found in Theorem 2 \cite{Kou02}. 
\subsubsection{Heston Stochastic Volatility Model} \label{annex:Heston}

With (\ref{HestonSDE}), we have

\begin{align}
c_1& = (1-e^{-\lambda T}) \frac{(\bar{u}-u_0)}{2 \lambda} - \frac{1}{2} \bar{u}T\\
c_2& = \frac{1}{8 \lambda^3} ( \eta T \lambda e^{- \lambda T}(u_0 - \bar{u}) 
(8 \lambda \rho - 4 \eta) + \lambda \rho \eta(1-e^{-\lambda T})(16\bar{u} - 8u_0)\\
&+ 2 \bar{u} \lambda T ( -4 \lambda \eta \rho  + \eta^2 + 4 \lambda^2)
+ 8 \lambda^2 (u_0 - \bar{u})(1-e^{-\lambda T})\\
& + \eta^2 ( (\bar{u} -2 u_0)e^{-2 \lambda T} + \bar{u} (6 e^{- \lambda T} -7) +2u_0 ) )
\end{align}

\section*{Acknowledgements}
We thank the reviewers and the associate editor for their constructive remarks to
improve the quality of this paper.
The authors would like also to thank C.W Oosterlee (Delft University of Technology)  and C. Necula (University of Zurich) for helpful comments.

\newpage

\bibliography{mybibliography}
\bibliographystyle{plain}

\end{document}